\documentclass{eptcs}

\usepackage{amsmath}
\usepackage{amssymb}
\usepackage{amsthm}
\usepackage{breakurl}

\usepackage{enumerate}
\usepackage{microtype}
\usepackage{framed}
\usepackage{hyperref}
\usepackage{comment}

\newtheorem{definition}{Definition}
\newtheorem{theorem}{Theorem}
\newtheorem{proposition}[theorem]{Proposition}

\newcommand{\cand}{\text{ and }}

\makeatletter
\newcommand{\raisemath}[1]{\mathpalette{\raisem@th{#1}}}
\newcommand{\raisem@th}[3]{\raisebox{#1}{$#2#3$}}
\makeatother

\newcommand{\bra}{\langle}
\newcommand{\ket}{\rangle}
\newcommand{\fall}[1]{{\forall\,{#1},\ }}
\newcommand{\fexist}[1]{{\exists\,{#1}\,{:}\ }}

\newcommand{\mb}[1]{{\bf #1}}
\newcommand{\mc}[1]{{\mathcal{#1}}}
\newcommand{\mf}[1]{{\mathfrak #1}}

\newcommand{\Mes}{\mathop{\mathrm{Mes}}}

\newcommand{\Vect}{\mathop{\mathrm{span}}}

\newtheorem{corollary}[theorem]{Corollary}

\specialcomment{smalltt}
    {\begingroup\ttfamily\footnotesize}{\endgroup}

\newcommand{\comp}{\mathrel{\mathrm C}}

\newcommand{\sas}{\mathbin{\&}}
\newcommand\sastar{\mathbin{\&^{\scriptstyle\#}}}
\newcommand{\bracc}[1]{\ensuremath{\left[\!\left[ {#1} \right]\!\right]}}

\makeatletter
\newcommand{\verif@es}{\blacktriangleright}
\newcommand{\verifies}{\mathrel{\verif@es}}
\newcommand{\verifiestar}{\mathrel{\verif@es^{\!\scriptstyle \#}}}
\newcommand{\verifiesm}[1]{\mathrel{\verif@es_{{\scriptstyle \mathfrak {#1}}}}}
\newcommand{\verifiestarm}[1]{\mathrel{\verif@es^{\!\scriptstyle \#}_{{\!\scriptstyle \mathfrak {#1}}}}}
\makeatother

\newcommand{\authorrunning}{Olivier Brunet}
\newcommand{\titlerunning}{Quantum Measurements from a Logical Point of View}

\begin{document}

\title{Quantum Measurements from a Logical Point of View}

\author{Olivier Brunet \\ \texttt{olivier.brunet at normalesup.org}}

\maketitle

\begin{abstract}
  We introduce a logic modelling some aspects of the behaviour of the measurement process, in such a way that no direct mention of quantum states is made, thus avoiding the problems associated to this rather evasive notion. We then study some properties of the models of this logic, and deduce some characteristics that any model (and hence, any formulation of quantum mechanics compatible with its predictions and relying on a notion of measurement) should verify. The main results we obtain are that in the case of a Hilbert space of dimension at least 3, using a strengthening of the Kochen-Specker theorem, we show that no model can lead to the certain prediction of more than one atomic outcome. Moreover, if the Hilbert space is finite dimensional, then we are able to precisely describe the structure of the predictions of any model of our logic. In particular, we show that all the models of our logic do exactly make the same predictions regarding whether a given sequence of outcomes is possible or not, so that quantum mechanics can be considered complete as long as the possibility of outcomes is considered.
\end{abstract}

As Jaynes puts it so vividly, ``\emph{our present [quantum mechanical] formalism is not purely epistemological; it is a peculiar mixture describing in part realities of Nature, in part incomplete human information about Nature -- all scrambled up by Heisenberg and Bohr into an omelette that nobody has seen how to unscramble}''~\cite{Jaynes:Scramble}.

One origin for theses difficulties is, in our opinion, the excessive reliance of the quantum mechanical formalism on the evasive notion of quantum state. Indeed, in most standard textbooks on quantum mechanics, the presentation of the theory starts with the postulate of its existence, and the rest of the exposition of the theory, including the important mechanism of measure, is based on this very notion. But, we insist, the existence of the quantum state is only postulated, so that the latter is an just abstract mathematical entity with no direct experimental counterpart, and hence doesn't have a clear status between being ontological or epistemological. This causes the aforementioned difficulties of interpretation which spread to other notions, such as that of measurement and the associated ``measurement problem''.

Yet, experimentally, the actual data that are obtained and dealt with are solely measurement outcomes and, in practice, any prediction regarding quantum mechanics is expressed in terms of measurement outcomes. This suggests, in order to gain a better understanding of quantum mechanics, to reverse the perspective by considering measurement outcomes as the primary component of the theory, instead of quantum states. This, of course, is not a new position regarding quantum mechanics. For instance, in \cite{Rovelli96RQM}, Rovelli states that ``one can take the view that outcomes of measurements are the physical content of the theory, and the quantum state is a secondary theoretical construction'' and refers back to Heisenberg and Bohr.
% @@@ Ontological models?

In this article, we introduce an approach using the tools of formal logic and based on %: We define an axiomatization of the some aspects of the behaviour of quantum measurements, and then sketch a study of the way some meaning can be assigned to this logical construction, in the form of a model \cite{Hodges:Shorter,Marker:Introduction,Marker:ModelTheory}. Formally, we consider statements of the form
a relation, denoted $ \Mes(s, \mc O, p, t) $, which intended meaning is that a system labelled by $s$ has been measured with an observable $\mc O$, yielding outcome $p$, and that $t$ labels the resulting system. 
%
% @@@@
%
% In the following, we will study the logic based on the ``$\Mes$'' relation.
%
We shall, for instance, consider the statement
\begin{equation}
\fall {s, \mc O, t} \Mes(s, \mc O, p, t) \implies \fexist {\mc O', u} \Mes(t, \mc O', p, u) \label{eq:song_zero}
\end{equation}
which indicates that outcome $p$ is always possible when measuring a system (labelled here by $t$) obtained by a measurement which itself yielded outcome $p$. Here, the part ``$\Mes(s, \mc O, p, t)$'' can be interpreted as a preparation phase, with $t$ labelling the prepared system, and ``$\fexist {\mc O', u} \Mes(t, \mc O', p, u)$'' is a prediction regarding this prepared system.

Another typical example is the following, where $p$ and $q$ are two mutually orthogonal outcomes:
\begin{equation}
\fall {s, \mc O, t} \Mes(s, \mc O, p, t) \implies \neg\bigl(\fexist {\mc O', u} \Mes(t, \mc O', q, u)\bigr) \label{eq:first_song}
\end{equation}
%In term of sequences of outcomes, this means that $(p, q)$ is impossible:
This corresponds to the fact that one cannot obtain two mutually orthogonal outcomes $p$ and $q$ in a row when measuring the same system. If this statement is true for any $q$ orthogonal to $p$, this impossibility can also be expressed by saying that if $t$ has been prepared by $\Mes(s, \mc O p, t)$, then $p$ is a certain outcome for~$t$: measuring it with an observable having $p$ as possible outcome will certainly yield outcome~$p$, since equation~\eqref{eq:first_song} states that any other outcome is impossible.

If we drop the labels refering to the systems and the observables, the previous two propositions can be interpreted as follows: the former corresponds to the fact that the sequence of outcomes $(p, p)$ is possible while the latter reflects the impossibility of the sequence $(p, q)$ as soon as $p$ and $q$ are orthogonal.

% @@@@

% For instance, given $s$, $\mc O$ and $p$, the statement ``$\fexist t \Mes(s, \mc O, p, t)$'' corresponds to the possibility of obtaining outcome $p$ when measuring $s$ with observable $\mc O$. This kind of statement will prove extremely important in the following, as it will be the basic way for expressing properties of quantum measurement. 

\medskip

Let us point out that we intentionally leave probabilities aside and only focus on the possibility or impossibility of obtaining a particular outcome. To that respect, an outcome $p$ is impossible (resp.\ certain, possible) if its probability equals $0$ (resp.\ equals $1$, is nonzero). Our approach can be seen as a \emph{possibilistic} one, using the term coined by Fritz \cite{Fritz:Possibilistic}.

\medskip

Our study will proceed as follows. First, we shall identify some properties involving the ``$\Mes$'' relation, and we will use quantum states so as to ensure that these properties are compatible with the prediction of quantum mechanics. However, we insist again on the fact that in the end, the obtained properties will be expressed using the ``$\Mes$'' construction only, so that there remains no reference whatsoever to quantum states. % This way, we will obtain a logical theory, i.e.~a collection of logical sentences which describe some aspects of the behavior of measurements and their outcomes.
We will then study the way one can assign some meaning, some semantics to these sentences, which is called a \emph{model} of the theory~\cite{Hodges:Shorter,Marker:Introduction,Marker:ModelTheory}. Obviously, the orthodox formulation of quantum mechanics based on quantum states will provide such a model but, more generally, any theory attempting to formalize quantum mechanics (and compatible with its predictions) and involving a notion of measurement shall lead to a model of this logical theory. Thus, the study of these models will allow us to address some questions such as whether the ``standard'' model (based on quantum states) has any particular status, or whether there exists other models leading to different possibilistic predictions.

% @@@ Bof, this last sentence. Instead, we study some properties of these models, and question to what point one is bound to consider quantum states.

\paragraph{A note on the formalism} In this article, we will only consider \emph{projective} measurements. Such a measurement can be written as
$\sum_i \lambda_i \Pi_i$ with $\sum_i \Pi_i = \mb 1_{\mc H}$ and $\fall {i,j} \Pi_i \Pi_j = \delta_{i,j} \Pi_i$, i.e.\ with each $\Pi_i$ being an orthogonal projector on some closed subspace. Moreover, following a logical approach, we will focus on the eigenspaces of such a projective measurement, rather on its eigenvalues.

Formally, given a Hilbert space $\mc H$ (or, more precisely, a Hilbert lattice $L(\mc H)$ consisting of the closed subspaces of $\mc H$) or, more generally, an \emph{orthomodular lattice} $L$, we define a (finite) \emph{observable} of $L$ as a finite subset $\mc O = \{p_1, p_2, \ldots, p_n\}$ of $L$ such~that
$$ \fall i p_i \neq \bot \qquad \quad \fall {i \neq j} p_i \leq {p_{\!j}}^\bot \qquad \hbox{and} \qquad \bigvee_{i = 1}^n p_i = \top $$
In the case where $L$ is the lattice $L(\mc H)$ associated to a Hilbert space $\mc H$, this definition corresponds to the set of eigenspaces of a Hermitian operator with finitely many eigenvalues. We invite the reader to refer to \cite{DallaChiara2001QuantumLogic,Svozil98Book,PlatoQuantLog} for further informations on these algebraic structures. In the following, $\mc M(L)$ will denote the set of finite observables of $L$.

\medskip

Furthermore, according to modern treatments of quantum mechanics (especielly in the field of quantum information and quantum computation), the most general kind of measurement one can perform on a quantum system is a P.O.\!V.M.~\cite{NielsenChung2000Book,Peres2002:QuantumTheory} and not a mere projective measurement. However, Neumark's theorem tells us that any P.O.\!V.M.\ can be realized by extending the used Hilbert space to a larger one, and then performing a projective measurement on this larger space. As a consequence, it is sufficient to only consider projective measurements in our formal study\footnote{Actually, the rigorous logical treatment of composite quantum systems and their formal counterparts goes far beyond the scope of this article, and will be developped in a future article.}.

\section{A Logic for Measurement}

Let us return to the statement that if a quantum system is measured twice in a row, one cannot obtain two mutually orthogonal outcomes, regardless which observables were measured, i.e. if $q \leq p^\bot$, then
\begin{equation} \label{eq:start_1}
\fall {s, \mc O, t} \Mes(s, \mc O, p, t) \implies \neg\bigl(\fexist {\mc O', u} \Mes(t, \mc O', q, u)\bigr)
\end{equation}

% @@@ Remind that we only consider projective measurements.

This obviously holds in orthodox quantum mechanics for projective measurements, as a consequence of the Born rule and the projection postulate. More precisely, if a system labelled by $s$ is in a state $|\varphi\ket$, and if $\Mes(s, p, t)$, then $\Pi_p |\varphi\ket \neq \mathopen| 0 \ket $ (otherwise outcome $p$ would not be possible) and the state $|\psi\ket$ of $t$ is colinear to $\Pi_p |\varphi\ket$. Now, since $q \leq p^\bot$, this implies that $\Pi_q \Pi_p = 0$, so that $\Pi_q |\psi\ket = \mathopen| 0 \ket $ and hence $q$ is not a possible outcome. Here, for $p \in L(\mc H)$, $\Pi_p$ obviously denotes the orthogonal projection on $p$.

We insist again on the fact that even though the justification of~\eqref{eq:start_1} relies on the notion of quantum states, it is stated in such a way that does not involve those: it is a statement regarding measurement outcomes only, and the previous justification ensures that it is consistent with the predictions of orthodox quantum mechanics.
This first property suggests the following definition:

\begin{definition}[Verification Statement]
For all $p \in L$ and $s \in S$, we define $s \verifies p$~by\footnote{A ``$\Delta$'' on top of an equality or an equivalence indicates a definition.}
$$ s \verifies p \stackrel \Delta \iff \neg\bigl( \fexist {\mc O, t} \Mes(s, \mc O, p^\bot, t)\bigr) $$
In that case, we will say that $s$ \emph{verifies} $p$.
\end{definition}
With this definition, statement~\eqref{eq:start_1} becomes
\begin{equation}
\fall {p, s, \mc O, t, q} \quad p \leq q \cand \Mes(s, \mc O, p, t) \implies t \verifies q \label{eq:start_2}
\end{equation}
Let us explore some other properties regarding measurements.
\paragraph{Valid Outcomes} The least element $\bot$ of $L$ (if $L$ is a Hilbert lattice, this corresponds to the nullspace $\{\mathopen|0\ket\}$) cannot be obtained as an outcome, that~is
$$ \fall {s, \mc O, t} \neg \Mes(s, \mc O, \bot, t). $$
Moreover, $\bot$ is \emph{the only} element of $L$ which cannot be obtained as an outcome:
$$ \fall p \bigl(\fall {s, \mc O, t} \neg \Mes(s, \mc O, p, t)\bigr) \implies p = \bot $$
Equivalently, %this statement can be expressed using verification statements~as:
these statements can be expressed respectively as $ \fall s s \verifies \top $ and $ \fall {p \neq \top} \fexist s \neg (s \verifies p) $.
%$ \fall s s \verifies \top $.

\paragraph{Measurability} Any system is likely to be measured, so that for every system $s$ and for every observable~$\mc O$, at least one outcome has to be possible, which we temporarily write~as
\begin{equation}
\fall {s, \mc O} \fexist {p, t} \Mes(s, \mc O, p, t). \label{eq:measurability_1}
\end{equation}

\paragraph{Possibilistic Noncontextuality} The next property can be expressed as follows: If an outcome $p$ is certain (resp.\ impossible) if measuring $s$ with a particular observable containing $p$, then it is also certain (resp.\ impossible) if measuring $s$ with \emph{any} observable containing $p$. In orthodox quantum mechanics, this observation follows from the Born rule, which states that the probability of obtaining an outcome $p$ when measuring a quantum system in normalized state $|\varphi\ket$~is $ \bra \varphi | \Pi_p | \varphi \ket $.
In particular, this probability does not dependent on which observable is actually measured, so that in our possibilistic approach, the impossibility or certainty of an outcome (corresponding to probabilities $0$ and $1$, respectively) is independent of the measured observable. We call this property \emph{possibilistic noncontextuality}, where \emph{noncontextuality} refers to the independence w.r.t.\ which observable is measured, and \emph{possibilistic} to the fact that we only consider the noncontextuality of certain and impossible properties. This should not be confused the stronger notion of noncontextuality usually associated to results such as the Kochen-Specker theorem. Let us derive some consequences from this property.

First, suppose that $p, q \in L$ are such that $p \leq q$, that $s \verifies p$ and let us consider observable $\mc O = \{p, q^\bot, p^\bot \wedge q\}$. With the assumption that $s \verifies p$, measuring $s$ with observable $\mc O$ would yield outcome $p$ with certainty (as both $q^\bot$ and $q \wedge p^\bot$ are lower that $p^\bot$) so that $q^\bot$ is not possible and, from noncontextuality, it is not possible to obtain $q^\bot$ as outcome when measuring $s$ (regardless of the measured observable), i.e.\ $s$ verifies $q$. We thus have shown~than
\begin{equation}
\fall {p \leq q} \fall s \quad s \verifies p \implies s \verifies q. \label{eq:contextuality_order}
\end{equation}

Suppose now that $p, q \in L$ are compatible, that $s \verifies p$ and $s \verifies q$ and consider observable $\mc O = \{p \wedge q, p^\bot, p \wedge q^\bot\}$. Measuring $s$ with $\mc O$ would yield $p \wedge q$ with certainty since $p^\bot$ is orthogonal to $p$ and is thus impossible as follows from $s \verifies p$, and $p \wedge q^\bot$ is also impossible, being orthogonal to $q$. But again, %from noncontextuality,
$p \wedge q$ is also certain if measuring $s$ with observable $\mc O' = \{p \wedge q, (p \wedge q)^\bot\}$ so that $(p \wedge q)^\bot$ is impossible w.r.t.\ $\mc O'$ and, using noncontextuality a last time, w.r.t.\ any observable. Thus, we~have
\begin{equation}
\fall {p \comp q} \fall s \quad s \verifies p \cand s \verifies q \implies s \verifies p \wedge q. \label{eq:compatible_meet}
\end{equation}
This property can be stated in a simpler way, using the Sasaki projection\footnote{The Sasaki projection is defined as $p \sas q \stackrel \Delta = q \wedge (p \vee q^\bot)$ and is the lattice-theoretical equivalent of the orthogonal projection.}. Indeed, if $s \verifies p$ and $s \verifies q$ (without the assumption that they are compatible), then $s \verifies p \vee q^\bot$ so that $s \verifies q \wedge (p \vee q^\bot)$ since $q$ and $p \vee q^\bot$ are compatible. As a consequence, equation~\ref{eq:compatible_meet} can equivalently be stated~as
$$ \fall {p, q \in L} \fall s \quad s \verifies p \cand s \verifies q \implies s \verifies p \sas q. $$

\medskip

Possibilistic noncontextuality allows us to reexpress some previous properties in a simpler way. For instance, using \eqref{eq:contextuality_order}, equation~\eqref{eq:start_2} can be replaced~by
$$ \fall {s, \mc O, p, t} \Mes(s, \mc O, p, t) \implies t \verifies p. $$
Similarly, measurability can be expressed~as $ \fall {s, \mc O} \neg\bigl(\fall {p \in \mc O} s \verifies p^\bot\bigr) $. But an observable $\mc O \in \mc M(L)$ is a finite collection of mutually orthogonal (and hence compatible) elements of $L$ so that, from equation~\eqref{eq:compatible_meet}, we deduce
$$ \bigl(\fall {p \in \mc O} s \verifies p^\bot\bigr) \iff s \verifies \bigwedge \bigl\{p^\bot \bigm| p \in \mc O \bigr\} \iff s \verifies \bot $$
% \begin{align*}
% \bigl(\fall {p \in \mc O} s \verifies p^\bot\bigr) & \iff s \verifies \bigwedge \bigl\{p^\bot \bigm| p \in \mc O \bigr\} \\
% & \iff s \verifies \Bigl( \bigvee \{ p \mid p \in \mc O\} \Bigr)^\bot \\
% & \iff s \verifies \top^\bot \\
% & \iff s \verifies \bot
% \end{align*}
Thus, measurability simply becomes: $ \fall s \neg(s \verifies \bot) $.

\paragraph{Compatible Preservation} We now present a last property, relating the verification of properties before and after a measurement. In order to express it, let us first translate verification statements in terms of quantum states. An outcome $p^\bot$ is impossible for a system in state $|\varphi\ket$ if its probability is $0$, that is, using the Born rule, if $\Pi_{p^\bot} |\varphi \ket = \mathopen| 0 \ket $. This means that if $s$ is in state $|\varphi\ket$, then that one~has\footnote{This ``translation'' will become clearer and more rigourous once we have introduced model-theoretic elements.} $ s \verifies p \iff |\varphi\ket \in p $.

Suppose then that $s \verifies p$ and $\Mes(s, \mc O, q, t)$ with $p$ and $q$ compatible. System $t$ is then in a state $|\psi\ket$ colinear to $\Pi_q |\varphi\ket$. But having $s \verifies p$ means that $|\varphi\ket \in p$ so that $\Pi_p |\varphi\ket = |\varphi\ket$, and the compatibility of $p$ and $q$ means that $\Pi_p$ and $\Pi_q$ commute. As a consequence,
$ \Pi_p \Pi_q |\varphi\ket = \Pi_q \Pi_p |\varphi\ket = \Pi_q |\varphi\ket $
and, similarly, $\Pi_p |\psi\ket = |\psi\ket$ so that $t \verifies p$. This shows that the following property is compatible with the predictions of orthodox quantum mechanics:
$$ \fall {p \comp q} \quad \fall {s, \mc O, t} s \verifies p \cand \Mes(s, \mc O, q, t) \implies t \verifies p $$
We call this \emph{compatible preservation} since the verification of an element $p \in L$ is preserved during a measurement, provided that its outcome $q$ is compatible with $p$. Again, considering possibilistic noncontextuality, this statement can be rewritten~as
$$ \fall {p, q} \quad \fall {s, \mc O, t} s \verifies p \cand \Mes(s, \mc O, q, t) \implies t \verifies p \sas q $$

\medskip

We summarize all these properties by defining the following logical theory:
\begin{definition}
Given an orthomodular lattice $L$, we define $\mc T_L$ as the theory consisting of the following axioms and axiom schemata:%\footnote{Any universal quantification on the elements of the orthomodular lattice corresponds to a axiom schema.}:
\begin{subequations}
\begin{align}
& \fall s s \verifies \top \label{eq:verifies_top} \\
& \fall s \neg(s \verifies \bot) \label{eq:not_verifies_bot} \\
\hbox{for all $p \neq \top$} \quad & \fexist s \neg (s \verifies p) \label{eq:is_outcome} \\
% \quad & \fall {s, \mc O, p, t} \Mes(s, \mc O, p, t) \implies t \verifies p \label{eq:noncontradiction} \\
\hbox{for all $p, q \in L$ such that $p \leq q$,} \quad & \fall s s \verifies p \implies s \verifies q \label{eq:noncontextuality_1} \\
\hbox{for all $p, q \in L$,} \quad & \fall s s \verifies p \cand s \verifies q \implies s \verifies p \sas q \label{eq:noncontextuality_2} \\
\hbox{for all $p, q \in L$,} \quad & \fall s \fall t s \verifies p \cand \Mes(s, q, t) \implies t \verifies p \sas q \label{eq:compatible_preservation} % \\[4pt]
\end{align}
\end{subequations}
\end{definition}
%
% @@@ Dumb It can be remarked that $\mc T_L$ can be seen as a one-sorted first-order theory since all the quantifications occur on the elements of $S$ (the universal quantifications on elements of $L$ only mean that there is a sentence for each possible assignement of $p$ and~$q$).
% \noindent It can be remarked that from equations~\eqref{eq:verifies_top} and~\eqref{eq:compatible_preservation}, it follows that $\Mes(s, \mc O, p, t) \implies t \verifies p$.

% @@@ Add the fact that $\Mes$ statement always appears as $\fexist{\mc O} \Mes(s, \mc O, p, t)$.

\section{Models}

All the axioms that constitute $\mc T_L$ are only sequences of characters, i.e.\ syntactical objects. They describe in mathematical terms some properties that ``$\Mes$'' should verify. In order to give assign a meaning to these sentences, we need to consider a structure made of a set $A$, called the \emph{universe} and a  relation $M$ reflecting the syntactical construction ``$\Mes(s, \mc O, p, t)$''. Let us first note that this construction actually only appears in the form ``$\fexist{\mc O} \Mes(s, \mc O, p, t)$'', so that $M$ can in fact be defined as a ternary relation on $A \times L \times A$, any explicit reference to an observable being unnecessary. It can also be remarked that such a structure $(A, M)$ can be seen as a labelled directed graph, with $A$ being its set of vertices, and with $(a, p, b) \in M$ denoting an arrow from $a$ to $b$ labelled by $p \in L$.

Intuitively, a graph $\mf G = (A, M)$ verifies a sentence $\varphi$ if and only if the graph verifies the translation of $\varphi$ in terms of $A$ and $M$. In this case, we write $\mf G \models \varphi$. For instance, one~has
$$ \mf G \models \fall {p \in L} \fall {s, t} \bigl(\fexist{\mc O} \Mes(s, \mc O, p, t)\bigr) \implies \fexist {\mc O', u} \Mes(t, \mc O', p, u) $$
if and only if it is true that $ \fall {p \in L} \fall {a, b \in A} M(a, p, b) \implies \fexist {c \in A} M(b, p, c)$.

\begin{definition}
A graph $\mathfrak G = (A, M)$ is a \emph{model} of $\mc T_L$ if and only if it verifies every axiom of $\mc T_L$. In that case, we write $ \mathfrak G \models \mc T_L$.
\end{definition}

While $\mc T_L$ is a set of syntactical elements, a model of $\mc T_L$ is an actual set equipped with an actual relation, i.e.\ an actual directed graph labelled by elements of $L$ in which the properties expressed by $\mc T_L$ do hold. With such a model, elements of $A$ can be seen as specifications of quantum states. To illustrate this, let us first introduce the model corresponding to the orthodox approach to quantum mechanics.

\begin{definition} [Hilbert Graph]
Given a Hilbert space $\mc H$, we define the \emph{Hilbert graph} $\mathfrak H_{\mc H} = (A_{\mf H}, M_{\mf H})$ by putting
$$ A_{\mf H} \stackrel \Delta = \bigl\{ |\varphi\ket \in \mc H \bigm| \bigl\| |\varphi\ket \bigr\| = 1 \bigr\} \quad \hbox{and} \quad M_{\mf H} \bigl(|\varphi\ket, p, |\psi\ket\bigr) \stackrel \Delta \iff \bigl\| \Pi_p | \varphi \ket \bigr\| \neq 0 \cand |\psi \ket = \frac {\Pi_p | \varphi \ket} {\bigl\| \Pi_p | \varphi \ket \bigr\|} $$
\end{definition}

\begin{proposition}
For every Hilbert space $\mc H$, $ \mathfrak H_{\mc H} \models \mc T_{L(\mc H)} $.
\end{proposition}
This result is the direct consequence of the fact that quantum states are actually the basic model of quantum mechanics, and that $\mc T_L$ has been defined considering quantum states. We recall that the verification relation~$\verifies$ translates in this model~as $ |\varphi\ket \verifiesm H p \iff |\varphi\ket \in p $.

However, elements of the universe of a model can also be uncomplete descriptions of a state. To illustrate this, we introduce another important model of~$\mc T_L$.

\begin{definition} [Lattice Graph]
Given an orthomodular lattice $L$, the \emph{lattice graph} $\mathfrak L_L= (A_{\mf L}, M_{\mf L})$ is defined by
% \begin{gather*}
% A_{\mf L} \stackrel \Delta = L^\star \qquad \hbox{where} \qquad L^\star \stackrel \Delta = L \setminus \{\bot\} \\
% M_{\mf L}(a,p,b) \stackrel \Delta \iff b \leq a \sas p
% \end{gather*}
$ A_{\mf L} \stackrel \Delta = L^\star$ (with $L^\star \stackrel \Delta = L \setminus \{\bot\}$) and $M_{\mf L}(a,p,b) \stackrel \Delta \iff b \leq a \sas p $.
%where $L^\star \stackrel \Delta = L \setminus \{\bot\}$.
\end{definition}
\begin{proposition}
Given an orthomodular lattice $L$, $ \mathfrak L_L \models \mc T_L $.
\end{proposition}
\begin{proof}
Let us first remark~that $a \verifiesm L p \iff a \leq p $.
% \begin{align*}
% a \verifiesm L p
% & \iff \neg\bigl(\fexist {b \in L^\star} b \leq a \sas p^\bot \bigr)\\
% & \iff a \sas p^\bot = \bot \\
% & \iff a \leq p
% \end{align*}
With this in mind, it is easy to prove that the different sentences of $\mc T_L$ do hold. For instance, sentences \eqref{eq:verifies_top} and~\eqref{eq:not_verifies_bot} simply state that $ \fall {a \in L^\star} \bot < a \leq \top$. %Sentence~\eqref{eq:is_outcome} translates into the statement $\fall {p \neq \top} \fexist{a \in L^\star} a \not \leq p$ which is verified by putting, for instance, $a = p^\bot$.
The other formulas can also be proven with no difficulties.
\end{proof}

The way $\mc T_L$ was designed -- with a single relation ``$\Mes$'' indicating the possibility of obtaining a specified outcome -- implies that any interpretation of quantum mechanics that includes a notion of measurement should, in order to comply with the predictions of orthodox quantum mechanics (corresponding in our approach to the model $\mathfrak H_{\mc H}$), lead to a model of $\mc T_L$ for some orthomodular lattice $L$. % In particular, this is the case for many hidden variable-based approaches to quantum mechanics.

Consider for instance the classical example introduced by Bell in \cite{Bell64}: a system, made of two spin one-half particles $A$ and $B$, is assumed to have its state completely specified by a parameter $\lambda$ belonging to some set $\Lambda$. Measuring the spin of particle $A$ along direction $\vec a$ yields outcome $A(\vec a, \lambda) \in \{+1, -1\}$ and, similarly, measuring $B$ along direction $\vec b$ yields outcome~$B(\vec b, \lambda)$.  Such an approach would provide a model $\mathfrak B$ of $\mc T_{L(\mb C^2 \otimes \mb C^2)}$, with $\Lambda$ as its universe. And even though the relation $M$ is yet unspecified, we can however express some verification statements: for all~$\vec a$ and~$\vec b$,
$$ \lambda \verifiesm B [\,A(\lambda, \vec a)] \otimes [B(\lambda, \vec b)] $$
where $[\,A(\lambda, \vec a)]$ (resp. $[B(\lambda, \vec b)]$) denotes the eigenspace corresponding to the indicated outcome.

Thus, the study of the models of $\mc T_L$ provides a general framework for understanding the meaning of the notion of quantum state and for determining the different ways of specifying (either completely or partially) the state of a quantum system.

\medskip

Before continuing, % with the study of some properties of the models of $\mc T_L$,
let us make a remark about the definition of our logic and, in particular, the property of possibilistic noncontextuality from which we have deduced equations~\eqref{eq:compatible_meet} and~\eqref{eq:noncontextuality_2}. The two models we have presented so far, $\mf H_{\mc H}$ and $\mf L_{\mc L}$, both verify
\begin{equation}
\fall {p, q \in L} \quad \fall {a \in A} a \verifiesm G p \cand a \verifiesm G q \implies a \verifiesm G p \wedge q \label{eq:wedge}
\end{equation}
i.e.\ both verify a version of equations~\eqref{eq:compatible_meet} where the compatibility requirement has been removed. Thus, one might wonder whether~\eqref{eq:wedge} could be used instead of~\eqref{eq:compatible_meet} and~\eqref{eq:noncontextuality_2}. Translated in terms of the ``$\Mes$'' relation, it states that if $a$ is such that if no outcome orthogonal to $p$ or $q$ can be obtained when measuring a quantum system in a state modelled by $a$, then no outcome orthogonal to $p \wedge q$ can be obtained either. But this is rather peculiar, since if $p$ and $q$ are not compatible, this relates the outcomes of necessarily \emph{distinct} observables, so that equation~\eqref{eq:wedge} seems rather \emph{ad hoc} and unnatural.

On the contrary, if we add the requirement that $p$ and $q$ are compatible, then we have seen during the justification of equation~\eqref{eq:compatible_meet} that one could deduce, considering a single observable, the certainty of outcome $p \wedge q$. 

However, we will show later that under some reasonable conditions which will be detailed later, equation~\eqref{eq:wedge} is necessarily verified by any model of $\mc T_L$. % In conclusion, possibilistic noncontextuality appears to be both a reasonable and well-grounded axiom for the definition of $\mc T_L$ and, combined with the other axioms, is sufficient for obtaining some rich structural properties for the models of $\mc T_L$, as we shall now~see.

\section{Sasaki Filters}

We shall now present some results regarding the set of those properties verified by an element of a model of $\mc T_L$. Formally, given such a model $\mf G = (A, M)$, we define the set of properties verified by $a \in A$~as
$$ \bracc a_{\mf G} = \bigl\{ p \in L \bigm| a \verifiesm G p \bigr\} $$
\noindent We can first remark that $\bracc a_{\mf G}$ is a \emph{consistent Sasaki filter} of $L$. We recall that given an orthomodular lattice, a \emph{Sasaki filter} \cite{Brunet07PLA,Brunet:2009} of $L$ is a subset $F \subseteq L$ verifying the following three properties:

\medskip

\begin{tabular}{rl}
1. & Non-empty: $F \neq \emptyset$ \\[2pt]
2. & Upper set: $p \in F \cand p \leq q \implies q \in F$ \\[2pt]
3. & $\sas$-Stability: $p, q \in F \implies p \sas q \in F$
\end{tabular}

\medskip

% \begin{enumerate}
% 	\item Non-empty: $F \neq \emptyset$
% 	\item Upper set: $p \in F \cand p \leq q \implies q \in F$
% 	\item $\sas$-Stability: $p, q \in F \implies p \sas q \in F$
% 	\item Consistency: $\bot \not \in F$
% \end{enumerate}
It is moreover \emph{consistent} if $\bot \not \in F$, i.e.\ if $F \neq L$. The verification of these properties by a $\bracc a_{\mf G}$ is a direct consequence of the definition of $\mc T_L$, and follow from axioms \eqref{eq:verifies_top}, \eqref{eq:not_verifies_bot}, \eqref{eq:noncontextuality_1} and~\eqref{eq:noncontextuality_2}. We also recall the following theorem from \cite{Brunet07PLA,Brunet:2009}, and an immediate consequence of~it:  
\begin{theorem} \label{theo:Kochen-Specker-Brunet} 
If $\mc H$ is a Hilbert space of dimension at least $3$, then every consistent Sasaki filter of $L(\mc H)$ contains at most one atom, i.e.\ one vector ray.
\end{theorem}

% Here, by \emph{atom}, we mean a $1$-dimensional subspace (that is, an element which is ``just'' above the least element $\bot = \bigl\{|0 \ket\bigr\}$ in $L(\mc H)$), and a Sasaki filter $F$ is consistent if it does not contain $\bot$, i.e.\ if~$F \neq L$. Clearly, if $\mf G = (A, M)$ is a model of $\mc T_{L}$ for some orthomodular lattice $L$, then for all $a \in A$, $\bracc a_{\mf G}$ is a consistent Sasaki filter of $L$. As a direct consequence, we~have:
\begin{corollary} \label{cor:ksb} 
If $\mathfrak G = (A, M)$ is a model of $\mc T_{L(\mc H)}$ with $\dim \mc H \geq 3$, then for all $a \in A$, $\bracc a_{\mathfrak G}$ contains at most one vector ray.
\end{corollary}

This result has important consequences. Recall the possibility of having a hidden-variable model of $\mc T_L$ as the one sketched earlier after Bell's article. The previous corollary simply forbids such a model, since it would lead to the presence of more than one atom in a consistent Sasaki filter. For instance, with the previous notations, given any two non-colinear and non-orthogonal vectors $\vec u$ and $\vec v$, we had for all~$\lambda \in \Lambda$,
$$ \lambda \verifiesm B [\, A(\lambda, \vec u)] \otimes [\, B(\lambda, \vec u)] \quad \hbox{and} \quad \lambda \verifiesm B [A(\lambda, \vec v)] \otimes [B(\lambda, \vec v)] $$
which is precisely ruled out by corollary~\ref{cor:ksb}. More generally, this result rules out models where there exists at least one element in the universe of the model which verifies two distinct atomic outcomes. This includes (but is not restricted to) models involving counterfactual definiteness, and constructions such as the one presented by Bell. % It is also worth mentionning that the interdiction of having such hidden variables models is independent of their possible locality, since this notion is not present in our approach.
We can also derive the following result, which implies the Kochen-Specker theorem~\cite{KochenSpecker67,PlatoKS}.
\begin{corollary} 
\label{corr:ks}
If $\mc H$ verifies $3 \leq \dim \mc H$, there is no model $\mf G = (A, M)$ of $\mc T_{L(\mc H)}$ such that there exists an element $a \in A$ % (and, \emph{a fortiori}, for all $a \in A$)
verifying
$$ \fall {\mc O \in \mc M\bigl(L(\mc H)\bigr)} \fexist {p \in \mc O} a \verifiesm G p $$
\end{corollary}
\noindent This shows that theorem~\ref{theo:Kochen-Specker-Brunet} is actually stronger that the Kochen-Specker theorem, since the assumption of noncontextuality has been replaced by possibilistic noncontextuality. %, the former being stronger than the latter.

Another important consequence relates to the position and momentum of a particle. Heisenberg's uncertainty principle teaches us that the position and the momentum of a particle cannot be known simultaneously. Theorem~\ref{theo:Kochen-Specker-Brunet} actually goes further, by stating that no model compatible with the predictions of orthodox quantum mechanics can simultaneously specify the position and the momentum of a particle, independently of whether these position and momentum are known (whatever is meant by \emph{being known}).

\section{A Representation Theorem in Finite Dimension}

% @@@ We need
% \begin{equation}
% \neg\bigl(\fexist {b \in A} \bracc a_{\mf G} \sastar p \subseteq \bracc b_{\mf G}\bigr) \implies p^\bot \in \bracc a_{\mf G} \label{eq:predictivity}
% \end{equation}

We now present another result regarding the models of $\mc T_{L(\mc H)}$ which applies to 
% which follows moreover from conditions~\eqref{eq:that_one} and~\eqref{eq:predictivity}. However, this result only applies to the Hilbert lattice $L(\mc H)$ associated to a
Hilbert spaces $\mc H$ such that $3 \leq \dim \mc H < \infty$. Let us start by some elementary bilinear algebra. For legibility reasons, elements of $L(\mc H)$ will be denoted using capital letters and vectors with lowercase letters even though previously, elements of an orthomodular lattice were denoted using lowercase letters. Moreover, for $P \in L(\mc H)$ (so that it is a closed subspace of $\mc H$), let $\Pi_P$ denotes the orthogonal projection on $P$.

\begin{proposition}
Given two closed subspaces $P$ and $Q$ of a Hilbert space, the restriction to $P$ of \/ $\Pi_P \circ \Pi_Q$, denoted $\Pi_P \circ \Pi_Q|_P$, is self-adjoint.
\end{proposition}
% \begin{proof}
% It is a well-known fact that the orthogonal projection on a closed subspace of a Hilbert space is a self-adjoint operator. We then have, for $u, v \in P$:
% \begin{align*}
% \bra u \mid \Pi_P \circ \Pi_Q (v) \ket & = \bra \Pi_P(u) | \Pi_Q(v) \ket \\
% & = \bra u \mid \Pi_Q(v) \ket \\
% & = \bra \Pi_Q(u) | v \ket \\
% & = \bra \Pi_Q(u) | \Pi_P(v) \ket \\
% & = \bra \Pi_P \circ \Pi_Q(u) | v \ket
% \end{align*}
% \end{proof}
% \medbreak
% 
As a consequence, if $\mc H$ is finite dimensional, which we will now assume, then $P$ admits an orthonormal basis made of eigenvectors of $\Pi_P \circ \Pi_Q |_P$. % Let $\{\alpha_i\}$ be such a basis, with corresponding eigenvalues $\{\lambda_i\}$, and let $P_\lambda$ denote the eigenspace of $\Pi_P \circ \Pi_Q|_P$ associated to eigenvalue~$\lambda$.
It directly follows~that: % The following two results are~direct.
%
% \begin{proposition} We~have
% $ P_1 = P \wedge Q $ and $ P_0 = P \wedge Q^\bot $.
% \end{proposition}
% \begin{proof}
% If $u \in P_1$, then $\bigl\| \Pi_Q(u) \bigr\| = \bigl\| u \bigr\|$ so that $\Pi_Q(u) = u$ and $P_1 \subseteq P \wedge Q$. Conversely, if $u \in P \wedge Q$, then $\Pi_P \circ \Pi_Q(u) = u$, so that
% $$ P_1 = P \wedge Q $$
% Now, if $u \in P_0$, then $\Pi_P\bigl(\Pi_Q(u)\bigr) = |0 \ket$, which means~that
% $$ \fall {v \in P} \bra v | \Pi_Q(u) \ket = 0 $$
% In particular, if $v = u$,
% $$ \bra \Pi_Q(u) | \Pi_Q(u) \ket = \bra u | \Pi_Q^2 (u) \ket = \bra u | \Pi_Q(u) \ket = 0 $$
% which implies that $u \in Q^\bot$. This shows that $P_0 \subseteq P \wedge Q^\bot$ and, obviously, if $u \in P \wedge Q^\bot$, then $\Pi_P \circ \Pi_Q(u) = \Pi_P( |0 \ket) = |0 \ket$ so that $P_0 = P \wedge Q^\bot$.
% \end{proof}
%
% \noindent A consequence of these equalities~is~that:
\begin{proposition}
Two subspaces $P$ and $Q$ are compatible if, and only if the spectrum of \/ $\Pi_P \circ \Pi_Q|_P$ verifies
$$ \mathrm{sp}(\Pi_P \circ \Pi_Q|_P) \subseteq \{0,1\}. $$
\end{proposition}
% \begin{proof}
% This follows directly from the previous proposition and the fact that $\Pi_P \circ \Pi_Q|_P$ is self-adjoint, so that $P = \bigvee \bigl\{ P_\lambda \bigm| \lambda \in \mb R \bigr\}$:
% \begin{align*}
% \hbox{$P$ and $Q$ are compatible} & \iff P = (P \wedge Q) \vee (P \wedge Q^\bot) \\
% & \iff P = P_1 \vee P_0 \\
% & \iff \fall {\lambda \not \in \{0,1\}} P_\lambda = \{\vec 0 \} \\
% & \iff \mathrm{sp}\bigl(\Pi_P \circ \Pi_Q|_P\bigr) \subseteq \{0,1\}
% \end{align*}
% \end{proof}
Let now $\mf G = (A, M)$ be a model of $\mc T_{L(\mc H)}$. If we define $\bracc a_{\mf G} \sastar Q = \bigl\{ P \sas Q \bigm| P \in \bracc a_{\mf G} \bigr\}$, it is clear from axiom~\eqref{eq:compatible_preservation} that if $M(a, P, b)$, then $\bracc a_{\mf G} \sastar P \subseteq \bracc b_{\mf G}$.
\begin{proposition} \label{prop:not_not_principal}
If $P$ and $Q$ are two distinct incompatible elements of $\bracc a_{\mf G}$, then neither $P$ nor $Q$ are minimal in~$\bracc a_{\mf G}$.
\end{proposition}
\begin{proof}
From the previous discussion, if $P$ and $Q$ are incompatible, then there exists an eigenvector $u$ of $\Pi_P \circ \Pi_Q|_P$ associated with an eigenvalue $\lambda \not \in \{0,1\}$. As a consequence, by defining $v = \Pi_Q(u)$, we have $u \not \in Q$ and $v \not \in P$. One can note moreover that $\Pi_P(v) = \lambda u$. Let us now define $C = \Vect(u,v)$. Having $\lambda \neq 0$, one can write
$$ C = \Vect(\lambda u, v - \lambda u) = \Vect\bigl(\Pi_P(v),v - \Pi_P(v)\bigr)$$
so that $C$ is compatible with $P$, having $\Pi_P(v) \in P$ and $v - \Pi_P(v) \in P^\bot$. This implies that $P \sas C = P \wedge C = \Vect\bigl(\Pi_P(v)\bigr) = \Vect(u)$.
Similarly, one can write $C = \Vect(v, u - v) = \Vect\bigl(\Pi_Q(u), u - \Pi_Q(u)\bigr)$ so that $C$ is compatible with $Q$ and thus $Q \sas C = \Vect(v)$.

As a consequence, $\bracc a_{\mf G} \sastar C$ contains two distinct vector rays -- namely $\Vect(u)$ and $\Vect(v)$ -- so that, following theorem \ref{theo:Kochen-Specker-Brunet}, there is no $b \in A$ such that $\bracc a_{\mf G} \sastar C \subseteq \bracc b_{\mf G}$. As a consequence, there is no $b \in A$ verifying $M(a,C,b)$, implying that $C^\bot \in \bracc a_{\mf G}$. Finally, since $\bracc a_{\mf G}$ is a Sasaki filter, it also contains
$ C^\bot \sas P$ so that $P$ is not minimal in $\bracc a_{\mf G}$:
$$ C^\bot \sas P = P \wedge (C^\bot \vee P^\bot) = P \wedge (C \wedge P)^\bot = P \wedge \bigl(\Vect(u)\bigr)^\bot < P $$
Similarly, $\bracc a_{\mf G}$ also contains $ C^\bot \sas Q $ which is strictly smaller than $Q$.
\end{proof}

\begin{corollary} \label{prop:at_least}
If $\mc H$ is a Hilbert space such that $3 \leq \dim \mc H < \infty$ and $\mf G = (A, M)$ is a model of $\mc T_{L(\mc H)}$, then for all $a \in A$, $\bracc a_{\mf G}$ contains at most one minimal element.
\end{corollary}

Now, since $\mc H$ is finite dimensional, $L(\mc H)$ has a finite height so that any non-empty subset of $L(\mc H)$ contains at least one minimal element. Combining this remark with the previous proposition, we obtain:

\begin{theorem} \label{corr:e_a}
Given a model $\mathfrak G = (A, M)$ of $\mc T_{L(\mc H)} $ where $3 \leq \dim \mc H < \infty$, for all $a \in A$, there exists an element $e(a) \in L(\mc H)$ such that $\bracc a_{\mathfrak G} = e(a)^\uparrow = \bigl\{ p \in L(\mc H) \bigm| e(a) \leq p \bigr\}$.
\end{theorem}
With this notation, it can be remarked that an outcome $p$ is \emph{certain} w.r.t. $a \in A$ if $e(a) \leq p$ and \emph{possible} if $p \not \leq e(a)^\bot$, and that if $M(a, p, b)$, then $e(b) \leq e(a) \sas p$.

An important consequence of this result is that, as discussed earlier, if $a \verifiesm G p$ and $a \verifiesm G q$, then necessarily $a \verifiesm G p \wedge q$ since $\bracc a_{\mf G}$ is such that $\bracc a_{\mf G} \leq p$ and $\bracc a_{\mf G} \leq q$ and hence $\bracc a_{\mf G} \leq p \wedge q$. Expressed another way, the previously discussed equation~\eqref{eq:wedge} is a theorem of $\mc T_{L(\mc H)}$ if $3 \leq \dim \mc H < \infty$.

Another consequence is the following: in \cite{RandallFoulis:Disaster}, Randall and Foulis called a ``metaphysical disaster'' the fact that in quantum logic, the physical properties of a quantum system were usually put in one-to-one correspondance with the set of experimental propositions, i.e.\ the closed subspaces of a Hilbert space. In our approach, this correspondance directly follows from theorem~\ref{corr:e_a}, since the set $\bracc a_{\mf G}$ of properties associated to an element $a$ of a model $\mf G$ can actually be represented by an element of the lattice $L(\mc H)$, namely by~$e(a)$.

% \begin{corollary} \label{corr:always_meet}
% If $\mc H$ is such that $3 \leq \dim \mc H < \infty$, then
% $$ \fall {p, q \in L(\mc H)} \quad \mc T_{L(\mc H)} \models \fall {s \in S} s \verifies p \cand s \verifies q \implies s \verifies p \wedge q $$
% \end{corollary}
% \begin{proof}
% Let $\mathfrak G = (A,M)$ be a model of $\mc T_{L(\mc H)}$, and let $a$ be such that $a \verifiesm G p$ and $a \verifiesm G q$. Since there exists an element $e(a) \in L(\mc H)$ such that $\bracc a_{\mathfrak G} = e(a)^\uparrow$, we have $e(a) \leq p$ and $e(a) \leq q$ so that $e(a) \leq p \wedge q$ and, finally, $a \verifiesm G p \wedge q$.
% \end{proof}

% This can be compared to equation~\eqref{eq:compatible_meet} where the deduction $s \verifies p \wedge q$ could only be made provided that $p$ and $q$ were compatible.

\section{Sequences of Measurement Outcomes} \label{sec:sequences}

We end this article by investigating the following question: given an orthomodular lattice $L$ and elements $p_1, \ldots, p_n \in L$, is it possible to obtain these elements as a sequence of successive outcomes when measuring a quantum system? In terms of $\mc T_L$, the question becomes whether
% \begin{multline*}
% \mc T_L \models \fexist{s_0, s_1, \ldots, s_n \in S} \Mes(s_0, p_1, s_1) \cand \\ \Mes(s_1, p_2, s_2) \cand \cdots \cand \Mes(s_{n-1},p_n,s_n)
% \end{multline*}
%
$$ \mc T_L \models \fexist{s_0, \mc O_1, s_1, \ldots, \mc O_n, s_n \in S} \Mes(s_0, \mc O_1, p_1, s_1) \cand \ \cdots \ \cand \Mes(s_{n-1},\mc O_n, p_n,s_n) $$
and we will investigate the verification of such a sentence by the models of $\mc T_L$. This suggests the following definition.
% This question can be easily translated in terms of the models of $\mc T_L$. % We first introduce this notion:
\begin{definition}
Given a complete orthomodular lattice $L$ and a model $\mf G = (A, M)$ of $\mc T_L$, we define the language $\ell(\mf G)$ of $\mf G$ as the set of labels of paths of $\mf G$. Formally, a word $\mb p = p_1 p_2 \cdots p_n$ on $L$ (i.e.\ a finite sequence of elements of $L$) is in $\ell(\mf G)$ if and only if there exists elements $a_0, a_1, \ldots, a_n \in A$ such~that $M(a_{k-1}, p_k, a_k)$ for all $k \in \bracc{1,n}$ (i.e. $\mb p$ is the label of path $(a_0, a_1, \ldots, a_n)$).
\end{definition}
In that case, given a model $\mf G$ and a word $\mb p = p_1 \cdots p_n$, the previous question amounts to determining whether $\mb p \in \ell(\mf G)$. The answer does \emph{a priori} depend on the model, which is extremely interesting as it could provide a method for discriminating models. Indeed, suppose that a word $\mb p = p_1 p_2 \cdots p_n$ can occur as a sequence of actual outcomes of a physical experiment, and that it does not belong to the language $\ell(\mf G)$ of some model $\mf G$. This would constitute a criteria for ruling out $\mf G$ as a correct model of quantum mechanics.
However, as we will see next, in the case of a Hilbert lattice $L(\mc H)$ with $3 \leq \dim \mc H < \infty$, all the models of $\mc T_{L(\mc H)}$ do actually define the same language. %, so that it is not possible to discriminate them this way. % We prove this by showing three successive inclusions:
% $$ \ell(\mathfrak G) \subseteq \ell(\mathfrak L_{L(\mc H)}) \subseteq \ell(\mf H_{\mc H}) \subseteq \ell(\mf G) $$

\begin{proposition}
For any model $\mathfrak G = (A, M)$ of $\mc T_{L(\mc H)}$ with $3 \leq \dim \mc H < \infty$, one~has
$ \ell(\mathfrak G) \subseteq \ell(\mathfrak L_{L(\mc H)}) $.
\end{proposition}
\begin{proof}
With the previous notation, if $M(a, p, b)$ then $e(b) \leq e(a) \sas p$. As a consequence, given $\mb p = p_1 p_2 \cdots p_n \in \ell(\mathfrak G)$, let $(a_0, \ldots, a_n)$ be a path in $\mf G$ labelled by $\mb p$, so that $\fall {k \in \bracc{1,n}} M(a_{k-1}, p_k, a_k)$. Then, for all $k \in \bracc{1,n}$, we have $ e(a_{k}) \leq e(a_{k - 1}) \sas p_k $ i.e.\ $M_{\mathfrak L}\bigl(e(a_{k-1}), p_k, e(a_k)\bigr)$. As a consequence, $\mb p$ is the label in $\mathfrak L_{\mc H}$ of the path $\bigl(e(a_0), \ldots, e(a_n)\bigr)$, so that~$\mb p \in \ell\bigl(\mathfrak L_{\mc H})$.
\end{proof}

Let us now recall that elements of $L(\mc H)$ are subspaces of $\mc H$, i.e.\ they are sets of vectors. In particular, for all $a, p \in L(\mc H)$, one has $ a \sas p = \bigl\{ \Pi_p |\psi\ket \bigm| |\psi\ket \in a \bigr\}$. 
% As a consequence, if $b \leq a \sas p$, then for all $|\varphi\ket \in b$ such that $\bra \varphi | \varphi \ket = 1$, there exists a $|\psi\ket \in a$ such that
% $$ \mathopen|\varphi\ket = \frac {\Pi_p | \psi \ket} {\bigl\| \Pi_p | \psi \ket \bigr\|}. $$
As a consequence, if $b \leq a \sas p$, then for all $|\varphi\ket \in b$, there exists a $|\psi\ket \in a$ such that $ \mathopen| \varphi \ket = \Pi_p | \psi \ket $.

\begin{proposition}
One~has $ \ell(\mathfrak L_{L(\mc H)}) \subseteq \ell(\mathfrak H_{\mc H}) $.
\end{proposition}

\begin{proof}
Let $\mb p = p_1 p_2 \cdots p_n \in \ell(\mathfrak L_{\mc H})$ and let $(e_0, e_1, \cdots e_n)$ be a path labelled by $\mb p$. Moreover, let $|\varphi_n\ket$ be a normalized element of $e_n$ and define backwards $|\varphi_{n - 1}\ket \in e_{n-1}$, \ldots, $|\varphi_1\ket \in e_1$ and  $|\varphi_0\ket \in e_0$ such that
$$ \fall {k \in \bracc{1,n}} \mathopen|\varphi_k \ket = \frac {\Pi_{p_k} | \varphi_{k-1} \ket} {\bigl\| \Pi_{p_k} | \varphi_{k-1} \ket \bigr\|} $$
In that case, for all $k$, $M_{\mathfrak H}\bigl(|\varphi_{k-1}\ket, p_k, |\varphi_k\ket\bigr)$, so that $\mb p$ labels the path $\bigl(|\varphi_0\ket, |\varphi_1\ket, \cdots, |\varphi_n\ket\bigr)$ in~$\mathfrak H_{\mc H}$.
\end{proof}

\begin{proposition}
For any model $\mathfrak G = (A, M)$ of $\mc T_{L(\mc H)}$ with $3 \leq \dim \mc H < \infty$, one~has
$ \ell(\mf H_{\mc H}) \subseteq \ell(\mathfrak G) $.
\end{proposition}
\begin{proof}
Let $\mb p = p_1 p_2 \cdots p_n$ be in $\ell(\mf H_H)$, and let $\bigl(|\varphi_0\ket, |\varphi_1\ket, \ldots, |\varphi_n\ket\bigr)$ be a path of $\mf H_{\mc H}$ labelled by $\mb p$. Since $\mf G$ is a model of $\mc T_{L(\mc H)}$ and $\Vect(|\varphi_0\ket) \neq \bot$, there exists $a_{-1}$ and $a_0$ in $A$ such that $M\bigl(a_{-1}, \Vect(|\varphi_0\ket), a_0\bigr)$. With the previous notations,
$$ \bot < e(a_0) \leq e(a_{-1}) \sas \Vect(|\varphi_0\ket) $$
But $\Vect(|\varphi_0\ket)$ is an atom of $L(\mc H)$ and $e(a_{-1}) \sas \Vect(|\varphi_0\ket) \leq \Vect(|\varphi_0\ket)$, so that
$ e(a_0) = \Vect(|\varphi_0\ket) $. Now, having $\Pi_{p_1} |\varphi_0\ket \neq \mathopen| 0 \ket $ or, equivalently, $p_1 \not \leq e(a_0)^\bot$, there exists $a_1 \in A$ such that $M(a_0, p_1, a_1)$. It~verifies $ e(a_1) \leq e(a_0) \sas p_1 $. But again, $e(a_0)$ is an atom, so that $e(a_0) \sas p_1$ is also an atom and~hence
$$ e(a_1) = e(a_0) \sas p_1 = \Vect(|\varphi_0\ket) \sas p_1 = \Vect\bigl(\Pi_{p_1} |\varphi_0\ket \bigr) = \Vect\bigl(|\varphi_1\ket\bigr) $$
Iterating this process, it is possible to define elements $a_1, \ldots, a_n$ such that each time, $M(a_{k-1}, p_k, a_k)$ and $e(a_k) = \Vect\bigl(|\varphi_k\ket\bigr)$. As a consequence, the path $(a_0, \ldots, a_n)$ is labelled by $\mb p$, so that $\mb p \in \ell(\mf G)$.
\end{proof}

To summarize these results, we have~shown:
\begin{theorem} \label{theo:one_language}
If \/ $3 \leq \dim \mc H < \infty$, then every model $\mf G$ of $\mc T_{L(\mc H)}$ verifies $ \ell(\mf G) = \ell\bigl(\mf L_{L(\mc H)}\bigr) = \ell(\mf H_H) $.
\end{theorem}

This theorem states that the languages defined by the models of $\mc T_{L(\mc H)}$ do only depend on $\mc H$, as soon as $3 \leq \dim \mc H < \infty$. If we denote this language $\ell(\mc H)$, we have, considering $\mf L_{L(\mc H)}$ and $\mf H_{\mc H}$ respectively:
\begin{align*}
\fall {p_1, \ldots, p_n \in L(\mc H)} \quad p_1 p_2 \cdots p_n \in \ell(\mc H) & \iff p_1 \sas p_2 \sas \cdots \sas p_n \neq \bot \\ & \iff \Pi_{p_n} \Pi_{p_{n-1}} \cdots \Pi_{p_1} \neq 0
\end{align*}
In particular, considering the Hilbert model $\mf H_{\mc H}$ alone, a sequence of outcomes is possible if and only if one can assign a quantum state to the starting system so that the quantum state of the last system (obtained by successive orthogonal projections) is different from the null vector. To that respect, regarding the possibility of sequences of outcomes, the description of a quantum system provided by quantum states is complete, thus partially answering a decades-old question~\cite{Einstein35EPR}.

\section{Conclusion and Perspectives}

In this article, we have introduced a logical formulation of quantum mechanics based solely on the apparent behavior of the measurement process. The obtained logic, called $\mc T_L$, and the study of its models do provide a general way to investigate properties of models of quantum mechanics. % the notion of \emph{admissible closure operator} on an orthomodular lattice $L$ which is closely related to the models of $\mc T_L$, as seen in propositions~\ref{prop:model_to_cl} and~\ref{prop:cl_to_model}.
For instance, in the case where the orthomodular lattice is the one associated to a Hilbert space $\mc H$ of dimension at least~$3$, we have shown in theorem~\ref{theo:Kochen-Specker-Brunet} that no model of $\mc T_{L(\mc H)}$ can have any of its elements verify more than one atomic property, thus ruling out a large class of hidden-variable models of quantum mechanics.

Moreover, if $\mc H$ is finite dimensional, it is possible to characterize precisely the structure of the collection of those properties verified by a state. A consequence of this result, studied in section~\ref{sec:sequences}, is that all the models of $\mc T_{L(\mc H)}$ make exactly the same predictions regarding whether a given sequence of outcomes is possible.

% To that respect, quantum states appear to be a convenient tool for determining the possibility of a sequence of outcomes: such a sequence is possible if and only if one can assign a quantum state to the starting system so that the quantum state of the last system (obtained by successive orthogonal projections) is different from the null vector. Obviously, in probabilistic terms, this can be interpreted as saying that the probability of this sequence is non zero. But theorem~\ref{theo:one_language} also teaches us that this method exactly captures \emph{all} the possible sequences of outcomes: there exists no model of $\mc T_{L(\mc H)}$ in which a given sequence is possible even though its probability using the Born rule is zero.

This article only initiates the study of the logic $\mc T_L$ and the ensuing approach of quantum mechanics. A first direction for future developments is the study of the possible extension of theorem~\ref{corr:e_a} to infinite dimensional Hilbert spaces. Another one would focus on the relation between the possibilistic approach developped in this article and the probabilistic one. Obviously, any probabilistic approach can lead to a possibilistic approach by considering whether the probability of an event is zero or one. But conversely, given a possibilistic model $\mf G = (A, M)$ of $\mc T_L$, how can one assign probabilities to the obtention of outcome $p$ in a state~$a \in A$?

A last important direction is the generalization of this approach to more complex settings. In the last part of this article, we have only considered sequences of outcomes, which correspond to a single quantum system being measured finitely many times in a row. An interesting and necessary generalization would be to consider directed acyclic graphs, corresponding to composite quantum systems. This would, in particular, provide a way to study entanglement and, hopefully, to help understand the notion of quantum state in a relativistic framework, with the vertices of a graph being associated to spacetime events.

\bibliographystyle{/Users/olivier/Library/texmf/bst/eptcs}
%\bibliography{/Users/olivier/Library/texmf/bibliography.bib}

\end{document}